\documentclass[10pt]{article}
\usepackage{preamble,comment}
\newcommand{\nc}{\newcommand}

\nc{\on}[1]{\operatorname{#1}}
\nc{\ol}[1]{\overline{#1}}
\nc{\wt}[1]{\widetilde{#1}}
\nc{\vp}{\varphi}
\nc{\eps}{\varepsilon}

\newcommand\restr[2]{{
  \left.\kern-\nulldelimiterspace
  #1 
  \vphantom{\big|}
  \right|_{#2}
}}

\nc{\R}{\mathbb R}
\nc{\C}{\mathbb C}
\nc{\Q}{\mathbb Q}
\nc{\Z}{\mathbb Z}
\nc{\N}{\mathbb N}
\nc{\F}{\mathbb F}
\nc{\D}{\mathbb D}
\nc{\ind}{\mathbf{1}}

\let\E\relax
\DeclareMathOperator{\E}{\mathbb{E}}

\nc{\Proj}{\operatorname{Proj}}
\nc{\Unif}{\operatorname{Unif}}

\nc{\CH}{\mathcal H}
\nc{\CL}{\mathcal L}
\nc{\CD}{\mathcal D}
\nc{\CZ}{\mathcal Z}
\nc{\CP}{\mathcal P}
\nc{\CF}{\mathcal F}
\nc{\CG}{\mathcal G}
\nc{\CN}{\mathcal N}
\nc{\CX}{\mathcal{X}}
\nc{\CY}{\mathcal{Y}}
\nc{\CU}{\mathcal{U}}
\nc{\CA}{\mathcal{A}}
\nc{\CE}{\mathcal{E}}
\nc{\CI}{\mathcal{I}}
\nc{\CC}{\mathcal{C}}
\nc{\CR}{\mathcal{R}}
\nc{\CS}{\mathcal{S}}
\nc{\CB}{\mathcal{B}}
\nc{\CO}{\mathcal{O}}

\nc{\defn}[1]{\textit{#1}}
\nc{\eg}{\emph{e.g.,~}}
\nc{\ie}{\emph{i.e.}}

\nc{\Sym}{\on{Sym}}

\newcommand{\hatket}[1]{
\ket{\smash{\hat{#1}}\vphantom{#1}}
}

\newcommand{\hatbraket}[2]{
\braket{\smash{\hat{#1}}\vphantom{#1}}{#2}
}
\newcommand{\brahatket}[2]{
\braket{#1}{\smash{\hat{#2}}\vphantom{#2}}
}

\nc{\clip}{\on{clip}}

\newcommand{\sddots}{\mathrel{\scalebox{0.6}{$\ddots$}}}
\newcommand{\poly}{\on{poly}}

\usepackage{color-edits}
\addauthor[Dutch]{dh}{blue}
\addauthor[to-do]{t}{red}

\title{Quantum 1-PCA with Pauli Measurements in \\Nearly Linear Time}
\author{
    Dutch Hansen\thanks{\texttt{jmhans13@cs.washington.edu}. This work was supported in part by an NSF Graduate Research Fellowship, and by the Institute for the Foundations of Machine Learning (IFML).} \and
    Jerry Li\thanks{\texttt{jerryzli@cs.washington.edu}. This work was supported in part by the Institute for the Foundations of Machine Learning (IFML).}
}
\date{
}

\begin{document}

\maketitle
\begin{abstract}

We consider the problem of quantum 1-PCA: given copies of an unknown $n$-qubit mixed state, recover a classical description of its leading eigenvector.
Our goal is to do so using non-adaptive and single-qubit measurements.
For an $n$-qubit state with top eigenvalue $\lambda$ and spectral gap at least $\Delta > 0$, we give an algorithm that recovers the leading eigenvector to fidelity at least $1 - \varepsilon$ with high probability using $\tilde{O}\left( {2^n \cdot \eta^2} / {\Delta^3 \varepsilon^3}\right)$ copies and $\tilde{O}((2^n/\Delta \varepsilon) \cdot \poly(\eta/\Delta\varepsilon))$ time, where $\eta = \max (1 - \lambda, \varepsilon)$.
All of our measurements are non-adaptively chosen, and performed in single-qubit Pauli bases.
When the spectral gap is constant and the desired accuracy is comparable to the noise level, i.e. $\varepsilon = \Omega (\eta)$, our runtime and copy complexity become $\tilde{O} (2^n / \varepsilon)$.
This generalizes the guarantees of~\cite{grewal2026nearly}, who achieved similar rates, but under the assumption $\eta = 0$, \emph{i.e.}, that the state was pure. 
Our results show that the same rates hold in the presence of state misspecification, up to polylogarithmic factors.

From a technical perspective, our algorithm works by recursively constructing low-dimensional subspaces that approximately preserve the target eigenvector.
To achieve nearly linear runtime dependence on the dimension of the Hilbert space, we develop a novel structured Pauli sampling scheme that enables fast batched computation of exponentially many projected Pauli matrices.
\end{abstract}



\section{Introduction}

A natural task in quantum learning is \emph{quantum PCA}, where the goal is to extract the best pure state approximation of an unknown $n$-qubit mixed state $\rho$, given measurement access to many copies of $\rho$.
For most metrics of interest, this is achieved by outputting the state which is the top eigenvector of $\rho$.
This is a well-studied problem in physics as well, where it is closely related to the problem of \emph{state purification}; see \emph{e.g.},~\cite{berthiaume1994stabilisation,barenco1997stabilization,werner1998optimal,peres1999error,cirac1999optimal,keyl2001rate,lloyd2014quantum,fu2016quantum,childs2025streaming}.
At a high level, this is because in many settings, we expect that the ``true'' state of interest should be some pure state.
However, due to issues such as model misspecification or noise in the state preparation, the resulting state we observe $\rho$ ends up being mixed with additional noise.
In such situations, the natural recovery question is not to recover $\rho$ itself, but rather its top eigenvector.

Formally, we consider the following setting: we are given copies of an $n$-qubit state $\rho$.
We assume that the top eigenvalue of $\rho$ is $\lambda$, and its second eigenvalue is $\lambda_2 < \lambda$, so that the top eigenvector of $\rho$, denoted $\ket{\psi}$, is unique.
We will also let $\Delta = \lambda - \lambda_2$ denote the spectral gap of the state.
Our goal is to output a state $\hatket{\psi}$ so that $\left|\hatbraket{\psi}{\psi}\right|^2 \geq 1 - \varepsilon$, for some target error $\varepsilon$.

By now, it is well understood how to solve this problem with access to arbitrary measurements.
Seminal work of~\cite{odonnell2017efficient} demonstrated that one can learn the optimal rank-1 approximation of an arbitrary mixed state to error $\varepsilon$ in $\tilde{O}(2^n / \varepsilon)$ copies, and it is also known that this is tight~\cite{hayashi1998asymptotic,BrussMacchiavello1999,scharnhorst2025optimal}.
That is, quantum PCA is possible with copy complexity which scales nearly linearly in the dimension of the Hilbert space.
However, the algorithm requires measurements which are heavily entangled both across the $n$ qubits of the state and across multiple copies of the state.
This renders the algorithm largely out of reach for existing or near-term quantum devices, except perhaps for very small-scale experiments.

On the other hand, another recent breakthrough work of~\cite{grewal2026nearly} demonstrated that one can match these statistical guarantees using single-qubit, single-copy measurements, a much more restricted but also realistically implementable class of measurements, but for the special case where the state $\rho$ is exactly pure.
That is, they show that if $\rho = \ketbra{\psi}$, there is an algorithm which uses $\tilde{O}(2^n / \varepsilon)$ copies and runtime, which only uses non-adaptively chosen single-qubit Pauli measurements, and which outputs a state $\hatket{\psi}$ which is $\varepsilon$-close to $\ketbra{\psi}$ in fidelity.
Note that in their setting, this is equivalent to learning the original state $\rho$ to infidelity $\varepsilon$.
However, their algorithm heavily leverages the fact that $\rho$ is exactly rank $1$, and to the best of our knowledge, fails to generalize even in the presence of very small amounts of noise.

This begs a natural question: can we obtain the best of both worlds?
That is, can we obtain an algorithm for quantum PCA that uses non-adaptively chosen, single-qubit Pauli measurements, but which works for general mixed states?
Stated simply, our central question is as follows:
\begin{center}
    {\it What is the complexity of quantum PCA with single-qubit measurements?}
\end{center}

Our main result is a new algorithm for quantum PCA with single-qubit measurements which still has copy complexity which is nearly linear in the dimension, and which also runs in nearly linear time in the dimension.
Up to polylogarithmic factors, our result fully generalizes the result of~\cite{grewal2026nearly}.
In fact, for a natural regime of parameters, our guarantees match (up to polylogarithmic factors) the copy complexity of quantum PCA with fully general measurements.
More concretely, we show:
\begin{theorem}
    Let $\rho$ be an unknown $n$-qubit quantum state with top eigenvalue $\lambda$, second eigenvalue $\lambda_2 < \lambda$, and top eigenvector $\ket{\psi}$.
    Let $\Delta = \lambda - \lambda_2$ denote the spectral gap of $\rho$.
    Let $\varepsilon, \delta > 0$, and let $\eta = \max(1 - \lambda, \varepsilon)$.
    Then, there is an algorithm, which, given $\Delta, \eta, \varepsilon >0$, and $N$ copies of $\rho$, where 
    \begin{align*}
        N = \tilde{O}\left( \frac{2^n}{\varepsilon}  \cdot \frac{\eta^2}{\Delta^3 \varepsilon^2} \cdot \log(1/\delta) \right) \; ,
    \end{align*}
    produces a pure state $\hatket{\psi}$ such that $\left|\hatbraket{\psi}{\psi}\right|^2 \geq 1-\varepsilon$ with probability at least $1-\delta$.
    Moreover, the algorithm performs only non-adaptive Pauli basis measurements and runs in time 
    \[
    \tilde{O}\left( \frac{2^n}{\varepsilon} \cdot \poly\left(\frac{\eta}{\Delta \varepsilon}  \right) \cdot \log(1/\delta) \right) \; .
    \]
\end{theorem}

\noindent Note that our full algorithm does not require exact knowledge of $\lambda$; see \cref{lem:pca-output,lem:pca-runtime}. In fact, we achieve the same guarantees with any $\Delta, \eta, \varepsilon > 0$ such that $\lambda \geq 1-\eta$ and $\lambda - \lambda_2 \geq \Delta$, with $\varepsilon \leq \eta$. We suppress this generality here for clarity of presentation.

We pause here to interpret this theorem. 
First, we are primarily interested in the regime where the spectral gap $\Delta$ is constant.
In this case, both our copy complexity and runtime are nearly linear in $d = 2^n$, which as mentioned above, is optimal. A key quantity which emerges from our theorem is the ratio $\eta / \varepsilon$.
Notice that $\eta$ is a natural measure of how close the state $\rho$ is to being pure, so one can interpret it as a measure of the signal to noise ratio.
Thus, the regime where $\eta / \varepsilon = \Theta (1)$ corresponds to the regime where our target accuracy is comparable to the mass of noise in the system. 
In this regime, observe that our theorem actually yields runtime and copy complexity that are $\tilde{O} (2^n / \varepsilon)$, matching the best possible rate.

In particular, in the special case where $\lambda = 1$, i.e. when the state $\rho$ is pure, we have that  $\eta / \varepsilon = \Theta (1)$.
Therefore, our result recovers both the copy and runtime complexity of~\cite{grewal2026nearly} for pure state tomography, and so our result is a strict generalization of their result to the more general mixed state setting.
Interestingly, as we describe now, while our algorithm shares a similar initial framework to the algorithm of~\cite{grewal2026nearly}, our techniques end up being quite different from theirs, even for the setting where $\rho$ is pure, and so our algorithm  yields a different route to achieving pure state tomography, which may be of independent interest.

\subsection{Our Techniques}\label{sec:overview}
At a high level, our algorithm proceeds qubit-wise.
For any prefix string $s \in \{0, 1\}^\ell$, let $\rho_s$ denote the unnormalized state obtained by measuring the first $\ell$ qubits of $\rho$ in the coordinate basis, and receiving outcome $s$, and similarly define $\ket{\psi_s}$.
Like in~\cite{grewal2026nearly},  we will reconstruct $\ket{\psi_s}$ ``backwards'': for some prefix $s \in \{0, 1\}^\ell$, for some $\ell < n$, we will seek to recover information about $\ket{\psi_s}$ based on approximate information recovered about the states with prefixes $s0$ and $s1$, i.e. the 1-bit extensions of $s$.
In~\cite{grewal2026nearly}, their algorithm seeks to directly reconstruct $\ket{\psi_s}$ using $\ket{\psi_{s0}}$ and $\ket{\psi_{s1}}$.
However, to do so, their algorithm crucially leverages the purity of $\rho$.
This  is because when $\rho$ is pure, it holds that $\rho_s = \ketbra{\psi_s}$ for all $s$. Thus, to recover $\ket{\psi_s}$, it suffices to recursively solve the same quantum PCA problem on a smaller scale.
However, this is no longer true when $\rho$ is mixed: in fact, after some basis measurements, the resulting post-measurement state can have a top eigenvector that is completely orthogonal to $\ket{\psi_s}$.
In particular, even if we are given the exact top eigenvectors of $\rho_{s0}$ and $\rho_{s1}$, we cannot in general hope to obtain a good approximation to $\ket{\psi_s}$.

Instead, our algorithm seeks to maintain, at every node $s$, a low-dimensional subspace $\Pi_s$ which approximately contains $\ket{\psi_s}$.
Our main algorithmic task is then to, given two subspaces $\Pi_{s0}$ and $\Pi_{s1}$, merge them into a single subspace $\Pi_s$ while maintaining good error and without blowing up the dimensionality of the subspace.
We do so in two steps.
First, consider the naively merged subspace with projector given by
\begin{align*}
    \tilde{\Pi}_s = \ketbra{0} \otimes \Pi_{s0} + \ketbra{1} \otimes \Pi_{s1}.
\end{align*}
It is straightforward to show that this projector would satisfy the necessary error bounds, but the dimension of the projector would also double, which would be prohibitive.
To address this issue, we interleave several classical $k$-PCA steps within each layer. 
Concretely, before the dimension of the naively constructed projections $\tilde{\Pi}_s$ grows too large, we use them, with appropriately chosen Pauli measurements, to build an estimate of $\rho_s$:
\begin{align*}
    \hat\rho_s \approx (\bra{s} \otimes \tilde{\Pi}_s) \rho (\ket{s} \otimes \tilde{\Pi}_s), \quad s \in \{0,1\}^\ell.
\end{align*} 
We then use these classical descriptions $\{\hat\rho_s\}_s$ to prune each projection $\tilde{\Pi}_s$ into a $k$-dimensional one $\Pi_s$, which we pass up to the next layer. Crucial to our nearly linear complexity guarantees is the observation that, in order to maintain most of the target state's energy, even amid noise in our estimate $\hat\rho_s$, $k$ need only be logarithmic in the dimension of $\ket\psi$ to give global error guarantees. 
Indeed, at any prefix, one can show that 
\begin{align*}
    \|(I-\Pi_s)\tilde{\Pi}_s \ket{\psi_s}\|_2^2 \lesssim \frac{\tr(\ketbra{s} \otimes I \nu)}{\lambda k} + \frac{\|\xi_s\|_\infty}{\lambda}.
\end{align*}
Summing over all prefixes at layer $\ell$, then over all layers, one may send the first term to $\varepsilon$ with $k = O(n)$, ignoring spectral terms. In this sense, our approach is intended to avoid prefix-wise accuracy guarantees which could propagate through the algorithm. To optimize our $\varepsilon$-dependence, we ultimately use a different energy inequality which allows us quadratic dependence on the perturbation operator norm, with a manageable tradeoff in the noise term; see~\Cref{lem:eigenspace-truncation-quadratic} for more details.

\subsubsection{Affine Pauli Sampling}

There is a final algorithmic challenge to implement this framework optimally.
With layer-$\ell$ projections $\{\tilde{\Pi}_s\}_s$, it's not immediately clear how to efficiently estimate all the unnormalized block estimates $\{\hat\rho_s\}_s$ using Pauli measurements. 
While it would be convenient to work directly with a basis for $ \Pi_s$, one notes that the projected Pauli matrices $\Pi_s \sigma_{x,z}\Pi_s$ also provide a frame for the subspace along which we need estimates. \begin{align*}
    (\bra{s} \otimes \Pi_s) \rho (\ket{s} \otimes \Pi_s) = \frac{1}{2^{n-\ell}} \sum_{x,z} \tr(\ketbra{s} \otimes \sigma_{x,z} \rho) \Pi_s \sigma_{x,z}\Pi_s.
\end{align*}
Hence, we can construct unbiased estimators for the projections of $\rho$ by post-selecting on prefix strings, measuring in a random Pauli basis, and using this measurement to sign the matrix $\Pi_s \sigma_{x,z} \Pi_s$. This projected Pauli frame is perhaps highly redundant. Importantly, however, it is only $O(k)$-dimensional. Indeed, we can choose these estimates to have vanishing variance with only $\tilde{O}(2^n k \cdot  \tr(\ketbra{s} \otimes I \rho))$ copies, so
$\tilde{O}(k 2^n)$ copies suffice to control all prefix estimates simultaneously.

If we want to make this algorithm efficient, a major issue remains: we cannot ever write down the full Pauli matrices explicitly. On one hand, the projected matrices $\Pi_s \sigma_{x,z} \Pi_s$ may be represented using only $O(k^2)$ complex numbers. In view of the Pauli matrices as phased permutations, it's also easy to check that each such matrix may be computed in $O(2^n k^2)$ time (at least in some basis for $ \tilde{\Pi}_s$). \begin{align*}
    \sigma_{x, z} \ket{y} = i^{x\cdot z} (-1)^{z \cdot y} \ket{y + x}, \quad x, y, z \in \F_2^n.
\end{align*} 
The issue, however, is that \emph{(i)} our subspaces $ \Pi_s$ (and hence the projected Paulis) are adaptively constructed, and \emph{(ii)} we need to access roughly $2^n$ of these matrices per layer, even if $k$ were kept constant. To fix this, we dispense with the natural uniform joint distribution over Pauli matrices. Instead, we sample from a more structured one which allows us to both simulate a uniform draw of $\approx 2^n$ Pauli matrices and compute all of the projected matrices at once. 

It turns out that such a structured distribution comes precisely from the commutation structure of the Pauli matrices. To see this, it is helpful to view our sample of Pauli matrices as a random subset of their phase space $F \subset \F_2^n \oplus \F_2^n$. And to simplify the situation, suppose we just have fixed unit vectors $\ket{u},\ket{v}$ and want to compute the collection of forms \begin{align*}
    \bra{u} \sigma_{x,z} \ket{v}, \quad (x,z) \in F.
\end{align*}
A first observation is that along some affine subspaces of $\F_2^n \oplus \F_2^n$, this mass computation can be made efficient.
Along the centered horizontal $H = \{(x, 0) : x \in \F_2^n\}$, for instance, one notices that the Paulis consist of translations on the computational basis. The forms of interest are therefore just evaluations of a certain convolution between $\ket{u}$ and $\ket{v}$: 
\begin{align*}
    \bra{u}\sigma_{x,0} \ket{v} = \sum_{y} \braket{u}{y+x} \braket{y}{v}.
\end{align*}
In fact, one can observe that by the Pauli commutation relations, any unitarily equivalent family must implement translation along some basis. Because the Paulis do so along the computational basis, it is an easy application of the fast Walsh-Hadamard transform algorithm to compute this form for all $x$ in $O(2^n \cdot n)$ time. 

Our basic strategy is to select $F$ as a random affine space along which this pre-computation may be done efficiently. Reminiscent of a basic hashing technique, a first attempt is to randomly tilt $H$ along the $Z$-axis, then shift. \begin{align*}
    F = \{(x, Ax + c)\}_{x \in \F_2^n}, \quad A \sim \Unif \F_2^{n \times n} , \, c \sim \Unif \F_2^n.
\end{align*}
First, it's not hard to see that the marginals $(x, Ax + c)$, $x \in \F_2^n$, are uniform over $\{x\} \times \F_2^n$ and pairwise independent. Second, because the chosen Paulis are precisely an affine transformation of $H$, we might hope to take advantage of its deterministic convolution structure. Fixing $x$, we expand one of the now-shifted Pauli matrices. \begin{align*}
    \sigma_{x, A x + c} &= \sum_y i^{x \cdot (A x + c)}(-1)^{y \cdot (A x + c)} \ketbra{y+x}{y}.
\end{align*}
The second observation is that if $A$ is forced to be symmetric, while the same two distributional properties remain, our affine space still admits an amenable decomposition; along $F$, the Pauli matrices now factor into a product of $\sigma_{x,0}$ with diagonals. Completing the square in the expression above, \begin{align*}
    \sigma_{x, Ax + c} &= i^{x \cdot (Ax+c) - x\cdot Ax} \begin{pmatrix}
        \sddots &&\\
        & i^{y \cdot A y}&\\
        &&\sddots 
    \end{pmatrix}\sigma_{x, 0} \begin{pmatrix}
        \sddots &&\\
        & (-i)^{y \cdot Ay}(-1)^{y \cdot c}&\\
        &&\sddots 
    \end{pmatrix}.
\end{align*}
Computing the forms of interest then reduces to a fast entry-wise pre-processing on $\ket{u}, \ket{v}$, followed by the classic Fourier procedure.

With pairwise independence, one again finds that $O(2^n)$ copies suffice for a layer, and this extends easily to the setting of general $k$ by considering $O(k^2)$ inner products of the form above. To mitigate copy complexity blowup in factors of $k$, we apply standard matrix estimation techniques; details are left to \cref{sec:subspace-tomography}.

\subsection{Related Work}

The task of quantum PCA, and more generally, the task of finding a good low-rank approximation to a given state, has a long history in physics and quantum computing.
Perhaps the first non-asymptotic guarantees for this problem were derived in~\cite{cirac1999optimal}, which showed optimal rates for purification of a single-qubit system. 
Several subsequent works~\cite{gross2010compressed,liu2011universal,flammia2012compressed,alquier2013rank,kueng2017low} derived estimators with non-asymptotic guarantees for many-qubit systems, several of which demonstrated very powerful connections between this task and the classical task of compressed sensing.
The first algorithms with nearly-optimal dimension dependence for this problem were derived in a sequence of works~\cite{haah2016sample,odonnell2016efficient,odonnell2017efficient}, culminating in the nearly-optimal rate in all parameters achieved in~\cite{odonnell2017efficient}.
Many works (including several of the aforementioned works) also obtain results for the related problem of tomography for low rank states, and in particular, for rank-1 states, see e.g.~\cite{hayashi1998asymptotic,haah2016sample,odonnell2016efficient,guta2020fast,van2023quantum,pelecanos2025mixed,scharnhorst2025optimal}.
One particularly relevant paper is~\cite{grier2024principal}, which specifically studies (among other questions) the task of finding the top eigenvector of a mixed state under a similar gapped assumption, but with general measurements.

In addition, there has been a parallel line of work developing algorithms for many learning tasks (including the one considered here), using restricted classes of measurements, such as single-copy Pauli measurements, including the aforementioned work of~\cite{grewal2026nearly}.
Several of the aforementioned papers use simple measurements such as Pauli measurements~\cite{gross2010compressed,liu2011universal,flammia2012compressed}, and subsequent work of~\cite{guta2020fast} refined these bounds.
A recent line of work has also studied the complexity of general state tomography with Pauli measurements which demonstrates that for full-rank states, Pauli basis measurements are provably worse than general single-copy measurements~\cite{acharya2025arenear,acharya2025arenot,yu2020sample}.
Other works have also proven lower bounds for these tasks with \emph{Pauli observable measurements}, a weaker class of measurements than general Pauli measurements~\cite{flammia2012compressed,lowe2025lower}, which demonstrate that such measurements cannot in general hope to obtain linear rates in the dimension.
Finally, we also mention that there is a large line of work on state certification and other quantum testing tasks using Pauli measurements or other classes of restricted measurements, see \emph{e.g.}~\cite{flammia2011direct,da2011practical,hayashi2015verifiable,pallister2018optimal,dangniam2020optimal,li2023robust,sater2025efficient,huang2025certifying,gupta2026few,li2026universal,coladangelo2026power,coladangelo2026robust}.
These papers do not seem to have much direct connection to the problem we study and the techniques we develop in our paper.

\section{Preliminaries}

Throughout, we take $i = \sqrt{-1}$; for clarity, the symbol $i$ is occasionally used as an index, though this is clear from context. For convenience throughout, we denote $\{0,1\}^0 = \{\emptyset\}$. For a matrix $\rho$, we let $\|\rho\|_p$, $p \in [1, \infty]$ denote the Schatten-$p$ norm. All logarithms are to base $e$, unless otherwise specified. For clarity of presentation, we omit numerical details regarding standard spectral computations in the real RAM model; one can handle such details deterministically, for instance via \cite{sobczyk2025deterministic}.

We let $\F_2$ denote the finite field on two elements, and $\on{Sym}_n(\F_2)$ the family of $n \times n$ self-adjoint linear maps over $\F_2^n$. We shall also make use of the Fourier transform over $\F_2^n$. For convenience, however, we will work with vector-valued maps. For $f, g : \F_2^n \to \C^k$ (or $\C^{k\times k}$), we take 
\begin{align*}
    \CF f(x) := \hat{f}(x) := \underset{y \sim \F_2^n}{\E} f(y)(-1)^{x \cdot y}.
\end{align*}
We also make use of the outer-product convolution \begin{align*}
    (f * g^\dagger) : \F_2^n \to \C^{k \times k}, \quad (f * g^\dagger)(x) = \underset{y \sim \F_2^n}{\E} f(y) g(x+y)^\dagger.
\end{align*}
In this form, the usual convolution identity may be written \begin{align*}
    \CF(f * g^\dagger)(x) = \hat{f}(x) \hat{g}(x)^\dagger.
\end{align*}

The $n$-site Pauli matrices consist of all $n$-fold tensor products of $X, Y, Z, I$. \begin{align*}
    X = \begin{pmatrix}
        &1\\
        1&
    \end{pmatrix}, \quad Y = \begin{pmatrix}
        &-i\\
        i&
    \end{pmatrix}, \quad Z = \begin{pmatrix}
        1&\\
        &-1
    \end{pmatrix}.
\end{align*}
Of particular interest is the family of qubit-wise unentangled measurements obtained by taking the 
tensor product of single-qubit Pauli eigenstates. For a Pauli $\sigma = \otimes_j \sigma_j$ with all non-identity components, the resulting rank-1 spectral projections may be written: \begin{align*}
    \bigotimes_{j=1}^n \left(\frac{I + (-1)^{y_j} \sigma_{j}}{2}\right), \quad y \in \F_2^n.
\end{align*}
We may alternatively parameterize the Pauli matrices by \begin{align*}
    \sigma_{x, z} := i^{x \cdot z} \bigotimes_{j=1}^n X^{x_j} Z^{z_j}, \quad x, z \in \F_2^n.
\end{align*}
Here, in the exponent of $i$, the dot product $x \cdot z$ is taken in $\C^n$ (or mod 4). In this form, one may view the Pauli matrices as signed permutations; a useful computational fact is the following. \begin{align}\label{eq:pauli-signed-permutation}
    \sigma_{x, z} \ket{y} = i^{x\cdot z} (-1)^{z \cdot y} \ket{y + x}.
\end{align}

\section{Fast Subspace Tomography with Pauli Measurements}\label{sec:subspace-tomography}

As a first ingredient, we require a subroutine which learns a small compression of $\rho$ along a known subspace. At the most basic level, we use only completeness of the Pauli matrices. Working directly with many randomly drawn Paulis can be expensive, however, and we access them along a low-dimensional subspace that evolves depending on the hidden state. To make our computations efficient, we leverage the $\F_2^n$ parameterization to sample certain well-structured Pauli matrices that can be efficiently constructed on the fly. We develop these tools concretely below.

\subsection{Lepskii's Method}

We find we can optimize $\varepsilon$-dependence by learning the matrices in operator norm. Here, we describe the estimator in generality before instantiating it in the subspace tomography algorithm. We require a Bernstein bound \cite{tropp2012user}.

\begin{fact}[Matrix Bernstein]\label{fact:bernstein}
    Consider a finite sequence $\{X_{j}\}$ of independent, Hermitian random matrices of dimension $d$. Assume that  each random matrix satisfies \begin{align*}
        \E X_{j} = 0, \quad  \|X_{j}\|_\infty \leq B \quad a.s., \quad\alpha^2 := \Big\| \sum_j \E X_{j}^2 \Big\|_\infty.
    \end{align*}
    Then for any $\delta > 0$, \begin{align*}
        \|\sum_j X_{j} \|_\infty \leq \frac{2}{3} B \log(d/\delta) + \sqrt{2\alpha^2 \log(2d/\delta)} \quad w.p. \,\geq 1-\delta.
    \end{align*}
\end{fact}

To control failure probabilities, we clip the spectrum of our empirical estimates to balance bias and variance appropriately. The ideal clipping depends on a possibly unknown variance, but we can use a standard adaptive technique to choose from a list of candidate clippings \cite{minsker2018annals,lepskii1992adaptive}. We clip the spectra of these matrices at several thresholds, then take empirical means.  
We now state the guarantee needed for our subspace tomography procedure (\cref{alg:subspace-tomography}). 

\begin{lemma}\label{lem:Lepskii}
    Consider an i.i.d. sequence $\{X_{j}\}_{j}$ of Hermitian matrices of dimension $d$, such that $\| \E X_{j}^2 \|_\infty \leq \beta^2$.
    Given access to scalars $\beta_{\on{max}}^2 \geq 8\beta_{min}^2 > 0$, \cref{alg:Lepskii} computes an estimate $\hat X$ using ${O(M \log(\beta_{max}^2 / \beta_{min}^2))}$ examples, and $O(M \poly(d) \log(\beta_{max}^2/\beta_{min}^2)^2)$ time, such that if $\beta_{max}^2 \geq \beta^2 \geq \beta_{min}^2$, then with probability at least $1-\delta$, \begin{align*}
        \|\hat X -\E X_{j}\|_\infty \lesssim \sqrt{\frac{\beta^2 \log(\beta_{max}^2/\beta_{min}^2 \cdot d/\delta)}{M}}.
    \end{align*}
\end{lemma}

\noindent We remark that one can actually replace $\beta_{max}^2/\beta_{min}^2$ with a logarithm in this quantity, up to constants in the estimate. For our purposes, the given claim will suffice. 

\begin{algorithm}[ht!]
\caption{Lepskii's Method}
\label{alg:Lepskii}

\setcounter{AlgoLine}{0}
\LinesNumbered
\everypar={\nl}

\KwIn{Parameters $d, \beta_{min}^2,\beta_{max}^2 (\geq 8\beta_{min}^2),  M$}
\KwIn{Sequence of \emph{i.i.d.} Hermitian matrices $X_j$ of dimension $d \times d$}

Set $I = \lceil \log_2(\beta_{max}^2/\beta_{min}^2)\rceil$

Organize $(I+1)M$ examples $\{X_{i,m}\}$ into $(I+1)$ blocks of $M$

\For{$i = 0, \ldots, I$}{
    Set $\beta_i^2 = 2^i \beta_{min}^2$

    Set $\mu_i = \sqrt{\frac{M\beta_i^2}{\log(4dI/\delta)}}$

    Record $\hat X_i = \frac{1}{M} \sum_m \clip_{\mu_i}(X_{i, m})$
}

Take the minimum $i_0 \in \{0, \ldots, I\}$ such that \begin{align*}
    \|\hat{X}_{i_0} - \hat{X}_{i}\|_\infty \leq 8\sqrt{\frac{\beta_i^2 \log(4dI/\delta)}{M}} \quad \forall i > i_0
\end{align*}

If condition is never satisfied, set $i_0 = 0$

Return $\hat X := \hat X_{i_0}$
\end{algorithm}

\noindent Here, we have taken the spectral extension of the scalar map \begin{align*}
    \clip_\mu(x) = \min\{\mu, \max\{x, -\mu\}\}, \quad x \in \R, \mu > 0.
\end{align*}
One should view $M$ as a variance-scaling parameter. We average over $M$ clipped estimates to achieve a desired scaling of errors.

\begin{proof}
    First note that for any indices $i,m$, \begin{align*}
        -\frac{X_{i,m}^2}{\mu_i} \preceq X_{i,m} - \clip_{\mu_i}(X_{i,m}) \preceq \frac{X_{i,m}^2}{\mu_i}.
    \end{align*}
    Moreover, note that for each $i$, \begin{align*}
        \|\E \sum_{m} (\clip_{\mu_i}(X_{i,m})- \E\clip_{\mu_i}(X_{i,m}))^2 \|_\infty \leq M \|\E \clip_{\mu_i}(X_{i,m})^2 \|_\infty \leq M\beta^2.
    \end{align*}
    Using \cref{fact:bernstein}, it holds with probability at least $1-\delta$ that for all $i$, \begin{align*}
        \|\hat{X}_i - \E X_{i,m}\|_\infty &\leq \|\hat{X}_i - \E \hat{X}_i\|_\infty + \|\E \hat{X}_i - \E X_{i,m}\|_\infty\\
        &\leq \frac{4\mu_i \log(4dI/\delta)}{3M} + \sqrt{\frac{2\beta^2 \log(4dI/\delta)}{M}} + \frac{\beta^2}{\mu_i}.
    \end{align*}
    
    Conditioned on the success of this event, the definition of $\mu_i$ implies that for every $\beta_i \geq \beta$, \begin{align*}
        \|\hat{X}_i - \E X_{i,m}\|_\infty &\leq 4 \sqrt{\frac{\beta_i^2 \log(4dI/\delta)}{M}}.
    \end{align*}
    Now consider the smallest $i_*$ such that $\beta_{i_*}^2 \geq \beta^2$, which must exist if $\beta^2 \leq \beta_{max}^2$. Note that for all $i \geq i_*$, a triangle inequality gives \begin{align*}
        \|\hat{X}_i - \hat{X}_{i_*}\|_\infty &\leq 8 \sqrt{\frac{\beta_i^2 \log(4dI/\delta)}{M}}.
    \end{align*}
    Hence, conditioned on success, $i_0 \leq i_*$, and another triangle inequality gives \begin{align*}
        \|\hat{X}_{i_0} - \E X_{i,m}\|_\infty \leq 16 \sqrt{\frac{\beta_{i_*}^2 \log(4dI/\delta)}{M}}.
    \end{align*}
    If also $\beta_{min}^2 \leq \beta^2$, then $\beta_{i_*}^2 = \beta_{min}^2 \leq \beta^2$ if $i_* = 0$ and $\beta_{i_*}^2 = 2\beta_{i_*-1}^2 \leq 2\beta^2$ otherwise. The claim follows.
\end{proof}

\subsection{Affine Pauli Estimators}

As described in \cref{sec:overview}, we efficiently compute our projected Pauli matrices by leveraging Fourier structure in our sampled Pauli matrices. We briefly formalize this step of the algorithm. Throughout the remainder of the section, we take $\tilde{U}_s$ as deterministic inputs to \cref{alg:projected-Pauli-precomputation,alg:subspace-tomography}. Later on, the adaptivity of these objects will be handled appropriately.

\begin{algorithm}[ht!]
\caption{Projected Pauli Pre-computation}
\label{alg:projected-Pauli-precomputation}

\setcounter{AlgoLine}{0}
\LinesNumbered
\everypar={\nl}

\KwIn{Parameters $n, \ell,k,R$}
\KwIn{Pairs $\{(A_r, c_r)\}_{r}$ of $\F_2$ symmetric matrices and vectors}
\KwIn{Partial unitaries $\{\tilde{U}_{s}\}_{s \in \{0,1\}^\ell}$ of dimension $2^{n-\ell} \times k_0$ ($k_0 \leq 2k$)}

Fix maps $u_{s}^{(r)}, v_{s}^{(r)} : \F_2^{n-\ell} \to \C^{k_0}$, for all $r,s$: 
\begin{align*}
    u_{s}^{(r)}(x) = i^{x \cdot A_r x} \tilde{U}_{s}^\dagger \ket{x}, \quad v_{s}^{(r)}(x) = i^{x \cdot A_r x} (-1)^{x\cdot c_r} \tilde{U}_{s}^\dagger \ket{x}
\end{align*}

Apply Fast Walsh-Hadamard to compute $\hat u_{s}^{(r)}, \hat v_{s}^{(r)}$, for all $r,s$

Apply Fast Walsh-Hadamard to compute for all $r,s$
\begin{align*}
    \pi_{s}^{(r)}(x) = i^{x \cdot (A_r x + c_r) - x \cdot A_r x} \cdot 2^{n-\ell} \CF^{-1}(\hat{u}_{s}^{(r)} {\hat{v}_{s}^{(r)\dagger}})(x), \quad x \in \F_2^{n-\ell}
\end{align*}

Return projected Pauli maps $\{\pi_{s}^{(r)}\}_{r,s}$

\end{algorithm}

\begin{lemma}\label{lem:projected-paulis}
    \cref{alg:projected-Pauli-precomputation} returns matrices \begin{align*}
        \pi_{s}^{(r)}(x) = \tilde{U}_{s}^\dagger \sigma_{x, A_r x + c_r} \tilde{U}_{s}, \quad s \in \{0,1\}^\ell, r \in [R].
    \end{align*}
\end{lemma}

\begin{proof}
    To start, note that \begin{align*}
        (x+y) \cdot A_r(x+y) &= x\cdot A_r + 2x \cdot A_r y + y \cdot A_r y \mod 4\\
        y \cdot (A_r x + c) &= y \cdot A_r x + y \cdot c \mod 2.
    \end{align*}
    Expand the Pauli matrix via \cref{eq:pauli-signed-permutation}. \begin{align*}
        \sigma_{x, A_r x + c_r} &= \sum_y i^{x \cdot (A_r x + c_r)}(-1)^{y \cdot (A_r x + c_r)} \ketbra{y+x}{y}\\
        &= i^{x \cdot (A_rx + c_r)-x\cdot Ax} \sum_y i^{(x+y)\cdot A_r(x+y)} (-i)^{y \cdot A_r y} (-1)^{y \cdot c_r} \ketbra{y+x}{y}.
    \end{align*}
    By linearity, \begin{align*}
        \tilde{U}_s^\dagger \sigma_{x, A_r x + c_r} \tilde{U}_s &= i^{x \cdot (A_r x + c_r) - x \cdot A_r x} \cdot 2^{n-\ell}(u_s^{(r)} * v_s^{(r)\dagger})(x)\\
        &= i^{x \cdot (A_r x + c_r) - x \cdot A_r x} \cdot 2^{n-\ell} \CF^{-1}(\hat{u}_s^{(r)} \hat{v}_s^{(r)\dagger})(x).
    \end{align*}
    The claim follows by definition of $\pi_{s}^{(r)}$.
\end{proof}

\begin{lemma}\label{lem:pauli-computation-runtime}
    \cref{alg:projected-Pauli-precomputation} runs in $O(2^n n^2 k^2 R)$ time.
\end{lemma}

\begin{proof}
    We may build each of the $2^\ell R$ maps $u_s^{(r)}, v_s^{(r)}$ as follows. Looking at each of the $2^{n-\ell}$ rows of $\tilde{U}_s$, we compute the phase factor in $O((n-\ell)^2)$ time using $A_r, c_r$, and multiply the $O(k)$ entries of this row point-wise. For each such map, the first Walsh-Hadamard transform may be done entrywise, on real and imaginary parts, for a total of $O(k(n-\ell)2^{n-\ell})$ time. These steps give at most $O(2^n n^2 k R)$ time.

    Next, the image of $\hat{u}_s^{(r)}\hat{v}_s^{(r)\dagger}$ is $k_0 \times k_0$ dimensional. Again computing each entry individually, we may compute the inverse transform of all outer products in $O(2^n R(n-\ell) k^2)$ time. Multiplying each projected Pauli by the diagonal phase factor takes $O(2^n (n-\ell)^2 k^2 R)$ time. The overall bound follows.
\end{proof}

\subsection{Subspace Tomography Algorithm}

We may now describe the full subspace tomography algorithm, which is called upon to simultaneously produce all estimates for prefixes at layer $\ell$. It executes $R$ different rounds, in each round producing matrices $\{\hat\rho^{(r)}_s\}_s$, which are then used in the Lepskii estimator to produce final layer-$\ell$ estimates $\{\hat\rho_s\}_s$. Within each round $r$, we draw a fresh affine space of pairs $\{(x, A_r x + c_r)\}_{x \in \F_2^{n-\ell}} \subset \F_2^{n-\ell} \times \F_2^{n-\ell}$, encoding the Paulis we will use to measure. Because our Paulis are not sampled uniformly, we require additional averaging within each round. 

Recall that in this section, we leave $\{\tilde{U}_s\}_s$ fixed; probabilistic estimates are computed accordingly. This will later be handled by conditioning. To handle the fact that our Lepskii estimates only provide guarantees within a prescribed variance range, we also introduce a parameter $p$, which is the minimum mass of unnormalized states for which we consider error in our matrix estimates. Below this threshold, we will incur losses which we ultimately control to sufficiently small error.

\begin{algorithm}[ht!]
\caption{Layer-$\ell$ Subspace Tomography}
\label{alg:subspace-tomography}

\setcounter{AlgoLine}{0}
\LinesNumbered
\everypar={\nl}

\KwIn{Parameters $n, \ell,k, M, p_\ell$}
\KwIn{Copies of unknown state $\rho$}
\KwIn{Partial unitaries $\{\tilde{U}_{s}\}_{s \in \{0,1\}^\ell}$ of dimension $2^{n-\ell} \times k_0$ ($k_0 \leq 2k$)}

Set variance thresholds $\beta_{min}^2 = 3k(p_\ell^2 + p_\ell/2^\ell)$, $\beta_{max}^2 = 24k$

Set $R = 2M \lceil \log_2(\beta_{max}^2/ \beta_{min}^2)\rceil$

Allocate $2^n R$ copies of $\rho$

Draw $A_r \overset{iid}{\sim} \Unif(\on{Sym}_{n-\ell}(\F_2))$ for $r\in [R]$

Draw $c_r \overset{iid}{\sim} \Unif(\F_2^{n-\ell})$ for $r\in [R]$

Use \cref{alg:projected-Pauli-precomputation} to compute projected Pauli matrices \begin{align*}
    \pi_{s}^{(r)}(x) = \tilde{U}_s^\dagger \sigma_{x, A_r x + c_r} \tilde{U}_s, \quad r\in [R], s \in \{0,1\}^\ell
\end{align*}

\For{$r =1, \ldots, R$}{

    \For{$x \in \F_2^{n-\ell}$}{

        \For{$j = 1, \ldots, 2^\ell$}{
            Draw fresh copy of $\rho$
        
            Measure first $\ell$ qubits in $Z$ basis; observe $s^{(r, x, j)} \in \{0,1\}^{\ell}$
        
            Measure final $n-\ell$ qubits in $\sigma_{x, A_r x + c_r}$ basis; observe $O^{(r, x, j)} \in \{\pm 1\}$
        }

        Record the estimates: \begin{align*}
            \hat\rho^{(r, x)}_{s} = \frac{1}{2^\ell} \sum_j \ind\{s^{(r, x, j)} = s\} \cdot O^{(r, x, j)} \pi_{s}^{(r)}(x), \quad s \in \{0,1\}^{\ell}
        \end{align*}
    }
    
    Record the round-$r$ estimates: \begin{align*}
    \hat\rho^{(r)}_{s} = \sum_x \hat\rho^{(r, x)}_{s}, \quad s \in \{0,1\}^\ell
\end{align*}
}

Run \cref{alg:Lepskii} on $\{\hat\rho_s^{(r)}\}_r$ for all $s$, with block size $M$, $\beta_{min}^2$, $\beta_{max}^2$;

Return Lepskii estimates $\{\hat\rho_s\}_s$

\end{algorithm}

\begin{lemma}
    For $s^{(r, x, j)}, O^{(r, x, j)}, A_r, c_r$ drawn as in \cref{alg:subspace-tomography}, \begin{align*}
        \E\big[ \sum_x \ind\{s^{(r, x, j)} = s\} \cdot O^{(r, x, j)} \pi_{s}^{(r)}(x) \big] &= (\bra{s} \otimes \tilde{U}_s^\dagger) \rho (\ket{s} \otimes \tilde{U}_s).
    \end{align*}
\end{lemma}

\begin{proof}
    For fixed $x$, $A_r x + c_r$ is uniform. Hence, \begin{align*}
        \E\big[\sum_x \ind\{s^{(r, x, j)} = s\} \cdot O^{(r, x, j)} \sigma_{x, A_r x + c_r} \big] = \frac{1}{2^{n-\ell}} \sum_{x, z} \tr(\ketbra{s} \otimes \sigma_{x, z} \rho) \sigma_{x,z} = (\bra{s} \otimes I) \rho (\ket{s} \otimes I).
    \end{align*}
    One concludes by \cref{lem:projected-paulis} and linearity.
\end{proof}

To apply the guarantees of Lepskii's method, we need only bound a second moment of our estimates.

\begin{lemma}\label{lem:second-moment}
     Using definitions from \cref{alg:subspace-tomography}, it holds that for any $r, s$, \begin{align*}
         \|\E {\hat{\rho}_{s}^{(r)}}{\hat{\rho}_{s}^{(r)}}\|_\infty \leq 3k \left( \tr(\ketbra{s}\otimes I \rho)^2 + \frac{\tr(\ketbra{s} \otimes I \rho)}{2^\ell} \right).
     \end{align*}
\end{lemma}

\begin{proof}
    For ease of notation, we fix $r$ and write $A_r = A$, $c_r = c$, $\hat\rho^{(r)}_s = \hat\rho$, and $\pi_s = \pi_s^{(r)}$. Expanding the moment, \begin{align*}
        \E \big[ \big(\hat\rho^{(r)}_s - \E[\hat\rho^{(r)}_s]\big)^2\big] = \E \E\big[ \big(\hat\rho - \E[\hat\rho \,|\, A, c]\big)^2 \,|\, A, c\big] + \E \big[ \big( \E[\hat\rho \,|\, A, c] - \E[\hat\rho] \big)^2\big].
    \end{align*}
    Note that \begin{align*}
        \E[\hat\rho \,|\, A, c] = \sum_x \tr(\ketbra{s} \otimes \sigma_{x, Ax + c} \rho) \pi_s(x).
    \end{align*}
    Since all $2^n$ measurement outcomes in a round are independent conditional to $A, c$, \begin{align*}
        \E\big[ \big(\hat\rho - \E[\hat\rho \,|\, A, c]\big)^2 \,|\, A, c\big] &= \frac{1}{2^\ell} \sum_x \E\big[ \big(\ind\{s^{(r, x, j)}=s\}O^{(r, x, j)} - \tr(\ketbra{s} \otimes \sigma_{x, Ax + c} \rho)\big)^2 \,|\, A, c\big] \pi_s(x)^2\\
        &\preceq \frac{1}{2^\ell} \sum_x \tr(\ketbra{s}\otimes I \rho) \pi_s(x)^2.
    \end{align*}
    Note that for fixed $x$, $Ax + c$ is uniform over $\F_2^{n-\ell}$. By the twirling identity, \begin{align*}
        \sum_x \E[\pi_s(x)^2] = \frac{1}{2^{n-\ell}} \tilde{U}_s^\dagger \left( \sum_{x, z} \sigma_{x,z} \tilde{U}_s \tilde{U}_s^\dagger \sigma_{x, z} \right)\tilde{U}_s =\tr(\tilde{U}_s \tilde{U}_s^\dagger) I_{k_0}.
    \end{align*}
    Hence,
    \begin{align*}
        \E \E\big[ \big(\hat\rho - \E[\hat\rho \,|\, A, c]\big)^2 \,|\, A, c\big] \preceq \frac{2k}{2^\ell}\tr(\ketbra{s} \otimes I \rho) .
    \end{align*}

    For the second term, we note that the strings $\{Ax + c\}_{x \in \F_2^{n-\ell}}$ are pairwise independent. \begin{align*}
        \E \big[ \big( \E[\hat\rho \,|\, A, c] - \E[\hat\rho] \big)^2\big] &= \sum_x \E \big( \tr(\ketbra{s} \otimes \sigma_{x, Ax + c} \rho) \pi_{s}(x) - \E[ \tr(\ketbra{s} \otimes \sigma_{x, Ax + c} \rho) \pi_{s}(x)]\big)^2\\
        &\preceq \sum_x \E \tr(\ketbra{s} \otimes \sigma_{x, Ax + c} \rho)^2 \pi_{s}(x)^2\\
        &\preceq 2k\tr(\ketbra{s} \otimes I \rho)^2.
    \end{align*}
    In the last line we again used the twirling identity. To conclude, use a triangle inequality and note that by Hermiticity, \begin{align*}
        \| \E[\hat\rho^{(r)}]^2\|_\infty= \|(\bra{s} \otimes \tilde{U}_s^\dagger)\rho (\ket{s} \otimes \tilde{U}_s)\|_\infty^2 \leq \tr(\ketbra{s} \otimes I \rho)^2.
    \end{align*}
\end{proof}

\begin{lemma}\label{lem:subspace-tomography-concentration}
    For $0 < p_\ell \leq 1/2$, and $k$ larger than an absolute constant, \Cref{alg:subspace-tomography} uses ${O(2^n n M \log(1/p_\ell))}$ copies of $\rho$, and returns Hermitian matrices \begin{align*}
        \hat\rho_s = (\bra{s}\otimes \tilde{U}_s^\dagger)\rho (\ket{s}\otimes \tilde{U}_s) + \xi_s,\quad s \in \{0,1\}^\ell
    \end{align*}
    such that with probability at least $1-\delta$, \begin{align*}
        \sum_{\substack{s \in \{0,1\}^\ell\\\tr(\ketbra{s} \otimes I \rho) \geq p_\ell}} \frac{\|\xi_s \|_\infty^2}{\tr(\ketbra{s} \otimes I\rho)} \lesssim \frac{k(\ell+1) \log(k/p_\ell \delta)}{M}.
    \end{align*}
\end{lemma}

\begin{proof}
    For $s \in \{0,1\}^\ell$, let \begin{align*}
        \beta_s^2 = 3k \left( \tr(\ketbra{s} \otimes I \rho)^2 + \frac{\tr(\ketbra{s} \otimes I \rho)}{2^\ell}\right).
    \end{align*}
    If $\tr(\ketbra{s} \otimes I \rho) \geq p_\ell$, then $\beta_{min}^2 \leq \beta_s^2 \leq \beta_{max}^2$. By \cref{lem:second-moment,lem:Lepskii}, it holds with probability at least $1-\delta$ that for all $s \in \{0,1\}^\ell$ such that $\tr(\ketbra{s} \otimes I \rho) \geq p_\ell$, \begin{align*}
        \|\hat\rho_s - (\bra{s} \otimes \tilde{U}_s^\dagger) \rho(\ket{s}\otimes \tilde{U}_s)\|_{\infty}^2 &\lesssim \frac{\beta_s^2}{M} \log\left( \frac{\beta_{max}^2}{\beta_{min}^2} \cdot \frac{2k \cdot 2^\ell}{\delta}\right)\\
        &\lesssim \frac{k(\ell+1) \log(k/p_\ell\delta)}{M} \left(\tr(\ketbra{s} \otimes I \rho)^2 + \frac{\tr(\ketbra{s} \otimes I \rho)}{2^\ell}\right).
    \end{align*}
    In the second line we used $p_\ell \leq 1/2$ and \begin{align*}
        \frac{\beta_{max}^2}{\beta_{min}^2} \leq \frac{8}{p_\ell^2 + {p_\ell}/{2^\ell}} \leq \frac{8 \cdot2^\ell}{p_\ell^2}.
    \end{align*}
    The claim follows by summing.
\end{proof}

\begin{lemma}\label{lem:subspace-tomography-runtime}
    For $p_\ell \in (0, 1/2]$, \cref{alg:subspace-tomography} runs in $O(2^n n^2 M \poly(k) \log(1/p_\ell)^2)$ time.
\end{lemma}

\begin{proof}
    By \cref{lem:pauli-computation-runtime}, the pre-computation of Paulis takes $O(2^n n^2 k^2 \cdot M \log(\beta_{max}^2/\beta_{min}^2))$ time. For each measurement round $r$, and $(n-\ell)$-bit string $x$, we maintain $2^\ell$-many $k_0 \times k_0$ matrix sums $\{\hat\rho_{s}^{(r, x)}\}_s$, which we add to on the fly. Iterating over $2^\ell$-many copies, we can add to a single matrix for each observation. This takes at most $O(2^n k^2)$ time in each round. At the end of each round, we sum estimates over all $x$, for each $s$, to produce the matrices $\{\hat\rho_s^{(r)}\}_s$; this takes $O(2^n k^2)$ time in each round. Each of the $2^\ell$ calls to \cref{alg:Lepskii} takes only $O(M \poly(k) \log(\beta_{max}^2/\beta_{min}^2)^2)$ time.
    Since $R = O(M \log(\beta_{max}^2/\beta_{min}^2))$, altogether, we have at most $O(2^n n^2 M \poly(k) \log(1/p_\ell)^2)$.
\end{proof}

\section{Full Algorithm}

\subsection{Projected Energy Inequalities}

The analysis hinges on two perturbation estimates, specialized to our needs. The first may be seen as a projected version of the classic Davis-Kahan $\sin \theta$ bound, and is perhaps the first hint that a dimension-reduction strategy is suitable. 

\begin{lemma}\label{lem:projected-davis-kahan}
    Let $\sigma \in \C^{d\times d}$ be a matrix of the form \begin{align*}
        \sigma = \lambda \ketbra{\psi} + \nu,
    \end{align*}
    for a unit vector $\ket{\psi}$, $\lambda > 0$, and $\nu \succeq 0$ satisfying a spectral tolerance \begin{align*}
        \bra{\psi}\nu \ket{\psi} = 0, \quad \lambda - \|\nu\|_\infty \geq \Delta > 0.
    \end{align*}
    Let $\xi$ be Hermitian, $\Pi$ be any orthogonal projection, and let $\hatket{\psi}$ be a top eigenvector of $\hat\sigma = \Pi \sigma \Pi + \xi$ (or an arbitrary unit vector if $\hat\sigma = 0$). Then it holds that \begin{align*}
        1 - \left|\hatbraket{\psi}{\psi}\right|^2 &\lesssim \frac{\lambda}{\Delta} \|(I-\Pi)\ket{\psi}\|_2^2+ \frac{\|\xi\|_\infty^2}{\Delta^2}.
    \end{align*}
\end{lemma}

\begin{proof}
    If $\Pi \sigma \Pi = 0$, the bound is immediate, and holds for any unit vector $\hatket{\psi}$. Similarly if $\hat\sigma = 0$, then \begin{align*}
        \|\xi\|_\infty = \|\Pi \sigma \Pi\|_\infty \geq \lambda \|\Pi \ket{\psi}\|_2^2 \geq \Delta \|\Pi \ket{\psi}\|_2^2.
    \end{align*}
    So the bound becomes vacuous, as \begin{align*}
        \frac{\lambda}{\Delta} \|(I-\Pi)\ket{\psi}\|_2^2+ \frac{\|\xi\|_\infty^2}{\Delta^2} \geq 1-\|\Pi\ket{\psi}\|_2^2 + \|\Pi \ket{\psi}\|_2^4 \geq 1/4.
    \end{align*}
    Henceforth, we assume $\Pi \sigma \Pi \neq 0$ and $\hat\sigma \neq 0$.
    
    Let $\ket{\phi}$ be a top eigenvector of $\Pi \sigma \Pi$. By positivity of $\nu$, we note $\lambda_1(\Pi \sigma \Pi) \geq \lambda \|\Pi \ket{\psi}\|_2^2$. Also, since the support of $\nu$ is orthogonal to $\ket\psi$, \begin{align*}
        \lambda_1(\Pi \sigma \Pi) \leq \bra{\phi} \sigma \ket{\phi} \leq \lambda \left| \braket{\psi}{\phi} \right|^2 + \|\nu\|_\infty (1 - \left|\braket{\phi}{\psi}\right|^2).
    \end{align*}
    Using both inequalities, we have \begin{align*}
        1 - \left|\braket{\phi}{\psi}\right|^2 \leq \frac{\lambda}{\Delta} \|(I-\Pi)\ket{\psi}\|_2^2.
    \end{align*}
    
    Next, by rank-1 interlacing, \begin{align*}
        \lambda_2(\Pi \sigma \Pi) \leq \|\Pi \nu \Pi\|_\infty \leq \|\nu\|_\infty.
    \end{align*}
    Hence, \begin{align*}
        \lambda_1(\Pi \sigma \Pi) - \lambda_2(\Pi \sigma \Pi) \geq \Delta - \lambda \|(I-\Pi)\ket{\psi}\|_2^2.
    \end{align*}
    By Davis-Kahan, \begin{align*}
        1 - \left| \hatbraket{\psi}{\phi} \right|^2 &\lesssim \left( \frac{\|\xi\|_\infty}{\Delta - \lambda \|(I-\Pi)\ket{\psi}\|_2^2}\right)^2 \lesssim \frac{\|\xi\|_\infty^2}{\Delta^2}.
    \end{align*}
    In the second inequality, we used the fact that if $\Delta < 2\lambda \|(I-\Pi)\ket{\psi}\|_2^2$, the global claim is immediate.
    Conclude by a triangle inequality on root infidelity then squaring.
\end{proof}

Next, we study the stability of our pruning procedure as we pass between layers. On a first pass, the following may be read with $t = \tr(\nu)$. For the purpose of our analysis, however, we will later need some slack in this parameter.

\begin{lemma}\label{lem:eigenspace-truncation-quadratic}
    Let $\sigma \in \C^{d\times d}$ be a matrix of the form \begin{align*}
        \sigma = \lambda \ketbra{v} + \nu,
    \end{align*}
    for a nonzero (possibly non-unit) vector $\ket{v}$, $\lambda > 0$, and $\nu \succeq 0$. Let $\hat\sigma = \sigma + \xi$
    for $\xi$ Hermitian. Let $\Pi$ be the orthogonal projection onto a top-$k$ eigenspace of $\hat\sigma$ (choosing an arbitrary top$-k$ eigenspace in case of degeneracy), for $k \leq d$. Then \begin{align*}
        \|(I-\Pi)\ket{v}\|_2^2 &\lesssim \frac{t}{\lambda k} + \frac{k\|\xi\|_\infty^2}{\lambda (\lambda \braket{v} + t)},
    \end{align*}
    for any $t \geq \tr(\nu)$.
\end{lemma}

\begin{proof}
    We bound the energy two ways.
    Using the spectral property of $\Pi$, \begin{align*}
        \lambda \bra{v} (I-\Pi) \ket{v}^2 &\leq \bra{v}(I-\Pi) \sigma (I-\Pi) \ket{v}\\
        &= \bra{v}(I-\Pi) \hat\sigma(I-\Pi)\ket{v} - \bra{v}(I-\Pi)\xi (I-\Pi)\ket{v}\\
        &\leq ( \lambda_{k+1}(\hat\sigma) + \|\xi\|_\infty) \bra{v}(I-\Pi)\ket{v}.
    \end{align*}
    By Weyl's inequality and rank-1 interlacing, \begin{align*}
        \lambda_{k+1}(\hat\sigma) &\leq \lambda_{k+1}(\sigma) + \|\xi\|_\infty\\
        &\leq \lambda_k(\nu) + \|\xi\|_\infty\\
        &\leq \frac{\tr(\nu)}{k} + \|\xi\|_\infty.
    \end{align*}
    Combining, we have \begin{align}\label{eq:linear-bound}
        \|(I-\Pi)\ket{v}\|_2^2 \leq \frac{\tr(\nu)}{\lambda k} + \frac{2\|\xi\|_\infty}{\lambda}.
    \end{align}
    
    Second, since $(I-\Pi)$ commutes with $\hat\sigma$ and $\sigma \ket{v} = \lambda \braket{v}\ket{v} + \nu \ket{v}$, \begin{align}\label{eq:spectral-property}
        \lambda \braket{v}(I-\Pi)\ket{v} &= (I-\Pi)\hat\sigma(I-\Pi)\ket{v} - (I-\Pi)\nu \ket{v} - (I-\Pi)\xi \ket{v}.
    \end{align}
    Note that \begin{align*}
        (I-\Pi)\nu (I-\Pi) \preceq (I-\Pi)\hat\sigma(I-\Pi) - (I-\Pi)\xi(I-\Pi).
    \end{align*}
    And hence,
    \begin{align*}
        \|(I-\Pi)\nu\|_\infty^2 &\leq \tr(\nu) 
        \cdot \|(I-\Pi)\nu(I-\Pi)\|_\infty \leq \frac{\tr(\nu)^2}{k} + 2\tr(\nu)\|\xi\|_\infty.
    \end{align*}
    Taking the inner product of \cref{eq:spectral-property} with the unit vector $\propto (I-\Pi)\ket{v}$, we observe \begin{align*}
        \lambda \braket{v} \|(I-\Pi)\ket{v}\|_2 &\leq \lambda_{k+1}(\hat\sigma) \|(I-\Pi)\ket{v}\|_2 + \|(I-\Pi)\nu\|_\infty \sqrt{\braket{v}} + \|\xi\|_\infty \sqrt{\braket{v}}\\
        &\leq \frac{\tr(\nu)}{k} \|(I-\Pi)\ket{v}\|_2 + \left(\left(\frac{\tr(\nu)^2}{k} + 2\tr(\nu)\|\xi\|_\infty \right)^{1/2} + 2\|\xi\|_\infty\right)\sqrt{\braket{v}}.
    \end{align*}
    That is, \begin{align}\label{eq:quadratic-bound}
        \left(\lambda \braket{v} - \frac{\tr(\nu)}{k} \right)\|(I-\Pi)\ket{v}\|_2 \leq \left(\left(\frac{\tr(\nu)^2}{k} + 2\tr(\nu)\|\xi\|_\infty \right)^{1/2} + 2\|\xi\|_\infty\right)\sqrt{\braket{v}}.
    \end{align}

    Now, if $t \geq \lambda \braket{v}/2$, then $t \geq (t + \lambda \braket{v})/4$. Then by \cref{eq:linear-bound} and Young's inequality, \begin{align*}
        \|(I-\Pi)\ket{v}\|_2^2 \leq \frac{(\tr(\nu)+t)}{\lambda k} + \frac{k\|\xi\|_\infty^2}{\lambda t} \leq \frac{2t}{\lambda k} + \frac{4k\|\xi\|_\infty^2}{\lambda(\lambda \braket{v} + t)}.
    \end{align*}
    Alternatively, if $t < \lambda \braket{v}/2$, we have \begin{align*}
        \lambda \braket{v} - \frac{\tr(\nu)}{k} \geq \frac{1}{2} \lambda \braket{v} \geq \frac{1}{4}\left( \lambda \braket{v} + t\right).
    \end{align*}
    By \cref{eq:quadratic-bound}, we have \begin{align*}
        \|(I-\Pi)\ket{v}\|_2^2 &\leq \frac{16 \braket{v}}{(\lambda \braket{v} + t)^2} \left(\left(\frac{\tr(\nu)^2}{k} + 2\tr(\nu)\|\xi\|_\infty \right)^{1/2} + 2\|\xi\|_\infty\right)^2\\
        &\leq \frac{32}{\lambda(\lambda \braket{v} + t)} \left(\frac{\tr(\nu)^2}{k} + 2\tr(\nu)\|\xi\|_\infty + 4\|\xi\|_\infty^2 \right)\\
        &\leq \frac{160}{\lambda(\lambda \braket{v} + t)} \left(\frac{\tr(\nu)^2}{k} + k \|\xi\|_\infty^2 \right)\\
        &\leq \frac{80 \tr(\nu)}{\lambda k} + \frac{160 k \|\xi\|_\infty^2}{\lambda (\lambda \braket{v} + t)}.
    \end{align*}
    In the third line we used Young's inequality. In the fourth, that $\tr(\nu) \leq t$.
\end{proof}

\subsection{Full Algorithm}

We may now state the full PCA algorithm and its guarantees. At any point where eigenvectors are chosen from a zero matrix, or the matrix is degenerate more generally, orthogonal unit vectors may be chosen arbitrarily; we address this in our arguments.

\begin{algorithm}[ht!]
\caption{Subspace Pruning}
\label{alg:pca}

\setcounter{AlgoLine}{0}
\LinesNumbered
\everypar={\nl}

\KwIn{Parameters $n, \Delta, \eta, \varepsilon, \delta$} 
\KwIn{Copies of unknown state $\rho$}

Set subspace dimension $k = \left\lceil n\eta / (\Delta \varepsilon)\right\rceil$

Set mass threshold $p = \varepsilon\Delta / n$

Set block size $M = O\left( \frac{n^5 \eta^2}{\Delta^3 \varepsilon^3} \log\left( \frac{n^3 \eta}{\Delta^2 \varepsilon^2 \delta} \right)\right)$

Initialize $1\times 1$ matrices $U_x = 1$, $x \in \{0,1\}^n$

\For{$\ell = n-1, \ldots, 1 $}{
    Set $\tilde{U}_s = \ketbra{0}\otimes U_{s0} + \ketbra{1} \otimes U_{s1}$, for $s \in \{0,1\}^\ell$
    
    \uIf{$2^{n-\ell} \leq k$}
    {
        Set $U_s = \tilde{U}_s$, for $s \in \{0,1\}^\ell$
    }
    \Else{
        Run \cref{alg:subspace-tomography} on $O(2^n n^2 M \log(1/p))$ fresh copies of $\rho$, at layer $\ell$, with mass threshold $p/2^\ell$, and with partial unitaries $\{\tilde{U}_s\}_{s \in \{0,1\}^\ell}$; denote returned matrices $\hat\rho_s$.

        Let $V_s$ be a matrix whose columns are top-$k$ eigenvectors of $\hat\rho_s$

        Set $U_s = \tilde{U}_s V_s$ for all $s$
    }
}

Set $\tilde{U}_\emptyset = \ketbra{0}\otimes U_{0} + \ketbra{1} \otimes U_{1}$

Run \cref{alg:subspace-tomography} at layer zero, on $O(2^n n M \log(1/p))$ fresh copies, and with the partial unitary $\tilde{U}_\emptyset$ to obtain a Hermitian matrix $\hat\rho_\emptyset$

Take the top unit eigenvector $\ket{v_\emptyset}$ of $\hat\rho_\emptyset$ and return $\hatket{\psi} = \tilde{U}_\emptyset \ket{v_\emptyset}$

\end{algorithm}

\begin{lemma}\label{lem:pca-output}
    Let $\Delta, \eta > 0$, and $\rho$ be an unknown $n$-qubit quantum state with top eigenvector $\ket\psi$, top eigenvalue (of multiplicity one) at least $1-\eta$, and spectral gap at least $\Delta$.
    Let $\varepsilon, \delta > 0$, with $\varepsilon \leq \eta$.
    Then, given $\Delta, \eta, \delta, \varepsilon$, and $N$ copies of $\rho$, where 
    \begin{align*}
        N = \tilde{O}\left( \frac{2^n \eta^2}{\Delta^3 \varepsilon^3} \log( 1/\delta)\right) \; ,
    \end{align*}
    \cref{alg:pca} produces a pure state $\hatket{\psi}$ such that $\left|\hatbraket{\psi}{\psi}\right|^2 \geq 1-\varepsilon$ with probability at least $1-\delta$. Here, we have suppressed polylogarithmic factors in $1/\varepsilon, 1/\Delta, \eta$, and $2^n$.
\end{lemma}

Throughout the evolution of \cref{alg:pca}, the matrices $\{\tilde{U}_s\}_{s \in \{0,1\}^\ell}$ are adaptively chosen. To invoke the guarantees for \cref{alg:subspace-tomography}, we condition on the full layer in our analysis.

\begin{proof}
    For $0 \leq \ell \leq n-1$, and $s \in \{0,1\}^\ell$, let $\ket{\psi_s} = (\bra{s} \otimes I_{2^{n-\ell}}) \ket{\psi}$, with $\ket{\emptyset} = 1$ for convenience. \begin{align*}
        \Pi_s = U_s U_s^\dagger, \quad \tilde{\Pi}_s = \tilde{U}_s \tilde{U}_s^\dagger, \quad r_s = \|(I - \Pi_s)\ket{\psi_s}\|_2^2.
    \end{align*}
    First observe that by definition and orthogonality, \begin{align*}
        \|(I-\tilde{\Pi}_s)\ket{\psi_s}\|_2^2 = r_{s0} + r_{s1}.
    \end{align*}
    Moreover, \begin{align}\label{eq:error-invariant}
        r_s &= \|(I-\tilde{\Pi}_s)\ket{\psi_s}\|_2^2 + \|(\tilde{\Pi}_s - \Pi_s)\ket{\psi_s}\|_2^2\nonumber\\
        &= r_{s0} + r_{s1} + \|(I-\Pi_s)\tilde{\Pi}_s \ket{\psi_s}\|_2^2.
    \end{align}
    For all $\ell$ such that $2^{n-\ell} \leq k$, it holds that $\tilde{\Pi}_s = \Pi_s$. Also note that $r_x = 0$ for all $x \in \{0,1\}^n$. Thus, if $\ell_{0} = \lceil n - \log_2(k)\rceil \geq 1$, expanding \cref{eq:error-invariant} gives \begin{align}\label{eq:errors-terminate}
        \sum_{s \in \{0,1\}^{\ell_0}} r_s = 0.
    \end{align}

    If $\ell_0 > 1$, consider some $1 \leq \ell < \ell_0$.
    Conditional to the full layer $\{\tilde{U}_s\}_{s \in \{0,1\}^\ell}$, \Cref{alg:subspace-tomography} returns estimates of the form \begin{align*}
        \hat\rho_s = (\bra{s} \otimes \tilde{U}_s^\dagger)\rho (\ket{s} \otimes \tilde{U}_s) + \xi_s.
    \end{align*}
    By construction, the columns of $V_s$ are an orthonormal system spanning a top-$k$ eigenspace of the matrix \begin{align*}
        \hat\rho_s &= \tilde{U}_s^\dagger (\bra{s} \otimes I)\rho (\ket{s} \otimes I) \tilde{U}_s + \xi_s\\
        &= \lambda \tilde{U}_s^\dagger \ketbra{\psi_s} \tilde{U}_s + (\bra{s}\otimes \tilde{U}_s^\dagger)\nu (\ket{s}\otimes \tilde{U}_s) + \xi_s
    \end{align*}
    We observe that \begin{align*}
       \tr((\bra{s}\otimes \tilde{U}_s^\dagger)\nu (\ket{s}\otimes \tilde{U}_s)) \leq \tr(\ketbra{s} \otimes I \nu).
    \end{align*}
    Also observe that \begin{align*}
        \tr(\ketbra{s} \otimes I \rho) = \lambda \bra{\psi_s} \tilde{U}_s\tilde{U}_s^\dagger \ket{\psi_s} + \lambda \|(I-\tilde{\Pi}_s)\ket{\psi_s}\|_2^2 + \tr(\ketbra{s} \otimes I \nu).
    \end{align*}
    Noting that $\hat\rho_s$ has dimension larger than $k$, and assuming $\tilde{U}_s^\dagger \ket{\psi_s} \neq 0$, we may combine with \cref{lem:eigenspace-truncation-quadratic} to obtain 
    \begin{align*}
        \|(I-V_sV_s^\dagger) \tilde{U}_s^\dagger \ket{\psi_s}\|_2^2 &\lesssim \frac{1}{k} \|(I-\tilde{\Pi}_s)\ket{\psi_s}\|_2^2 + \frac{\tr(\ketbra{s} \otimes I \nu)}{\lambda k} + \frac{k \|\xi_s\|_\infty^2}{\lambda \tr(\ketbra{s} \otimes I \rho)}.
    \end{align*}
    Moreover, 
    $(I-\Pi_s)\tilde{\Pi}_s = \tilde{U}_s(I-V_sV_s^\dagger)\tilde{U}_s^\dagger$, and hence
    \begin{align}\label{eq:eigenspace-truncation-quadratic-application}
        \|(I-\Pi_s)\tilde{\Pi}_s \ket{\psi_s}\|_2^2 \lesssim \frac{1}{k} \|(I-\tilde{\Pi}_s)\ket{\psi_s}\|_2^2 + \frac{\tr(\ketbra{s} \otimes I \nu)}{\lambda k} + \frac{k \|\xi_s\|_\infty^2}{\lambda \tr(\ketbra{s} \otimes I \rho)}.
    \end{align}
    In the case where $\tilde{\Pi}_s \ket{\psi_s} = 0$, the bound is vacuous.
    One may expand \cref{eq:error-invariant} once more; letting $c$ denote an absolute constant, we have 
    \begin{align*}
        \|(I-\tilde{\Pi}_\emptyset)\ket{\psi}\|_2^2 &= \sum_{\substack{s \in \{0,1\}\\\tr(\ketbra{s}\otimes I \rho) < p/2}} r_s +  \sum_{\substack{s \in \{0,1\}\\\tr(\ketbra{s}\otimes I \rho) \geq p/2}} r_s\\
        &\leq \sum_{\substack{s \in \{0,1\}\\\tr(\ketbra{s}\otimes I \rho) < p/2}} r_s + \sum_{s \in \{0,1\}^2} r_s + \sum_{\substack{s \in \{0,1\}\\\tr(\ketbra{s}\otimes I \rho) \geq p/2}} \|(I-\Pi_s)\tilde{\Pi}_s\ket{\psi_s}\|_2^2\\
        &\leq \left(1 + \frac{c}{k}\right) \sum_{\ell=1}^2 \sum_{\substack{s \in \{0,1\}^\ell\\\tr(\ketbra{s}\otimes I \rho) < p/2^\ell}} r_s + \left(1 + \frac{c}{k} \right) \sum_{\substack{s \in \{0,1\}^2\\\tr(\ketbra{s}\otimes I \rho) \geq p/4}}r_s \\
        &\hspace{8em} + c\sum_{\substack{s \in \{0,1\}\\\tr(\ketbra{s}\otimes I \rho) \geq p/2}} \left(\frac{\tr(\ketbra{s} \otimes I \nu)}{\lambda k} + \frac{k \|\xi_s\|_\infty^2}{\lambda \tr(\ketbra{s} \otimes I \rho)}\right).
    \end{align*}
    By induction and \cref{eq:errors-terminate}, \begin{align*}
        \|(I-\tilde{\Pi}_\emptyset)\ket{\psi}\|_2^2 &\leq \left(1 + \frac{c}{k}\right)^{\ell_0 - 1} \sum_{\ell=1}^{\ell_0} \sum_{\substack{s \in \{0,1\}^\ell\\\tr(\ketbra{s}\otimes I \rho) < p/2^\ell}} r_s \\
        &+ c\left(1 + \frac{c}{k}\right)^{\ell_0-2} \sum_{\ell=1}^{\ell_0-1} \sum_{\substack{s \in \{0,1\}^\ell\\\tr(\ketbra{s}\otimes I \rho) \geq p/2^\ell}}\left(\frac{\tr(\ketbra{s} \otimes I \nu)}{\lambda k} + \frac{k \|\xi_s\|_\infty^2}{\lambda \tr(\ketbra{s} \otimes I \rho)}\right)\\
        &\leq ce^{cn/k} \left( \frac{np}{\lambda} + \frac{n\eta}{\lambda k} + \frac{k}{\lambda} \sum_{\ell=1}^{\ell_0 -1} \sum_{\substack{s \in \{0,1\}^\ell\\\tr(\ketbra{s}\otimes I \rho) \geq p/2^\ell}} \frac{\|\xi_s\|_\infty^2}{\tr(\ketbra{s} \otimes I \rho)} \right).
    \end{align*}
    In the second inequality, we used $r_s \leq \braket{\psi_s} \leq \tr(\ketbra{s} \otimes I \rho)/\lambda$.
    If alternatively $\ell_0 \leq 1$, the left-hand side is zero and the bound still holds.

    Provided that $\|(I- \tilde{\Pi}_\emptyset) \ket{\psi}\|_2^2 < 1$, which will hold by our choice of parameters, it also holds that $\tilde{\Pi}_\emptyset \ket\psi \neq 0$. At the zeroth layer, consider the $2^n \times 2^n$ matrix 
    \begin{align*}
        \hat\rho &= \tilde{U}_\emptyset \hat\rho_{\emptyset} \tilde{U}_\emptyset^\dagger - \|\xi_\emptyset\|_\infty (I-\tilde{\Pi}_\emptyset)\\
        &= \tilde{\Pi}_\emptyset \rho \tilde{\Pi}_\emptyset + \tilde{U}_\emptyset \xi_\emptyset \tilde{U}_\emptyset^\dagger - \|\xi_\emptyset\|_\infty (I-\tilde{\Pi}_\emptyset).
    \end{align*}
    Note that by construction, $\hatket{\psi} = \tilde{U}_\emptyset \ket{v_\emptyset}$ is a top eigenvector of $\hat\rho$, and \begin{align*}
        \big\|\tilde{U}_\emptyset \xi_\emptyset \tilde{U}_\emptyset^\dagger - \|\xi_\emptyset\|_\infty (I-\tilde{\Pi}_\emptyset)\big\|_\infty \leq \|\xi_\emptyset\|_\infty.
    \end{align*}
    So by \cref{lem:projected-davis-kahan}, \begin{align*}
        1 - \left|\brahatket{\psi}{\psi}\right|^2 &\lesssim \frac{\lambda}{\Delta}\|(I - \tilde{\Pi}_\emptyset) \ket{\psi}\|_2^2 + \frac{\| \xi_\emptyset \|_\infty^2}{\Delta^2}\\
        &\lesssim e^{cn/k} \left( \frac{np}{\Delta} + \frac{n \eta}{\Delta k} + \frac{k}{\Delta} \sum_{\ell=0}^{\ell_0-1}\sum_{\substack{s \in \{0,1\}^\ell\\\tr(\ketbra{s}\otimes I \rho) \geq p/2^\ell}}  \frac{\|\xi_s\|_\infty^2}{\tr(\ketbra{s} \otimes I \rho)} \right)\\
        &\lesssim \frac{np}{\Delta} + \frac{n\eta}{\Delta k} + \frac{k}{\Delta} \sum_{\ell=0}^{\ell_0-1}\sum_{\substack{s \in \{0,1\}^\ell\\\tr(\ketbra{s}\otimes I \rho) \geq p/2^\ell}}  \frac{\|\xi_s\|_\infty^2}{\tr(\ketbra{s} \otimes I \rho)}.
    \end{align*} 
    In the second and third lines, we used $k \geq 1/\Delta \lor n$. 

    By \cref{lem:subspace-tomography-concentration} and conditioning, it holds with probability at least $1-\delta$ that for all $\ell \in \{0, \ldots, n-1\}$, \begin{align*}
        \sum_{\substack{s \in \{0,1\}^\ell\\\tr(\ketbra{s} \otimes I \rho) \geq p/2^\ell}} \frac{\|\xi_s \|_\infty^2}{\tr(\ketbra{s} \otimes I\rho)} \lesssim \frac{k n^2 \log(n^2 k/(\Delta \varepsilon \delta))}{M}.
    \end{align*}
    By our choice of $p$ and $k$, \begin{align*}
        1 - \left|\brahatket{\psi}{\psi}\right|^2 &\lesssim \varepsilon + \frac{k}{\Delta} \sum_{\ell=0}^{\ell_0-1}\sum_{\substack{s \in \{0,1\}^\ell\\\tr(\ketbra{s}\otimes I \rho) \geq p/2^\ell}}  \frac{\|\xi_s\|_\infty^2}{\tr(\ketbra{s} \otimes I \rho)}\\
        &\lesssim \varepsilon + \frac{k^2 n^3 \log(n^2 k/(\Delta \varepsilon \delta))}{\Delta M}. 
    \end{align*} 
    The claim follows by choice of $M$ in \cref{alg:pca}.
\end{proof}

\begin{lemma}\label{lem:pca-runtime}
    \cref{alg:pca} runs in time \begin{align*}
        \tilde{O}\left( \frac{2^n}{\Delta \varepsilon} \cdot \poly\left(\frac{\eta}{\Delta\varepsilon}\right) \cdot \log(1/\delta)\right)\; ,
    \end{align*} where we have suppressed logarithmic factors in the quantities $1/\varepsilon, 1/\Delta, \eta$, and $2^n$.
\end{lemma}

\begin{proof}
    Within each iteration $\ell$, constructing all the temporary bases $\tilde{U}_s$, $s \in \{0,1\}^\ell$ from old ones takes at most $O(2^{n} k)$ time. 
    The choice of $M$, $p$ and \cref{lem:subspace-tomography-runtime} imply that the call to \cref{alg:subspace-tomography} runs in
    \begin{align*}
        O(2^n n^2 M \poly(k) \log(2^\ell/p)^2) = \tilde{O}\left( \frac{2^n}{\Delta \varepsilon} \cdot \poly\left(\frac{\eta}{\Delta\varepsilon}\right) \log(1/\delta)\right).
    \end{align*} Taking all top-$k$ eigenbases may be done in $2^{\ell} \poly(k)$ time, and computation of all $U_s$ in $2^{n} k^2$ time. The overall runtime is as claimed.
\end{proof}



\newpage

\clearpage
\subsection*{Acknowledgments}

We thank William He for enlightening conversations at the beginning of this project. 

\subsection*{AI Disclosure}

Multiple technical steps and optimizations were derived in discussion with GPT-5.6 and GPT-6. 
The idea of iteratively maintaining subspaces was first proposed by the authors, and the details were then worked out in conversations with GPT.
The structured Pauli sampling scheme was initially proposed by GPT-5.6, but was later refined and clarified by the authors. 
The authors wrote the entirety of the present manuscript, and claim responsibility for any errors and omissions.

\bibliographystyle{alpha}
\bibliography{refs}




\end{document}